\documentclass{amsart}
\usepackage{amssymb}

\usepackage{graphicx}
\usepackage{amscd}

\newtheorem{theorem}{Theorem}
\theoremstyle{plain}

\newtheorem{corollary}{Corollary}

\newtheorem{definition}{Definition}
\newtheorem{example}{Example}

\newtheorem{lemma}{Lemma}

\numberwithin{equation}{section}

\input{tcilatex}

\begin{document}

\begin{center}
\bigskip {\Huge All quantum expectation values as classical statistical mean
values.}

\bigskip

{\LARGE Antonio Cassa}

antonio.cassa@tin.it

http://xoomer.virgilio.it/cassa\_lazzereschi/antonio/welcome.htm

\bigskip

\bigskip {\Large Abstract}
\end{center}

Given a physical quantum system described by a Hilbert space\textit{\ }$%
\mathcal{H}$, for any bounded quantum observable (a bounded self-adjoint
operator) $T$ it is possible to define several ''hidden observable''
functions $f:\mathcal{H\rightarrow }\mathbb{R}$ associated to $T$ and for
any quantum mixed state (a density matrix) $D$ it is possible to define
several ''hidden mixed states'' (probability measures) $\mu $ on $\mathcal{H}
$ associated to $D$ in such a way that the following equality is verified: 
\begin{equation*}
Trace\left[ b(T)\cdot D\right] =\int_{\mathcal{H}}b(f(\psi ))\cdot d\mu
(\psi )
\end{equation*}

whatever is the continuous function $b:\mathbb{R\rightarrow R}$.

This formula gives a general way to express any expectation value computable
in a quantum theory as a classical statistical mean value.

\bigskip

\section{\protect\Large Introduction}

This article is a mathematical paper giving another way to express all the
expectation values and all the probabilities of a quantum theory but, in the
same time, is a hidden variable theory avoiding all known no-go theorems.

For a better comprehension of our results we suggest a list of hypothesis to
keep in mind; in the following we will constantly refer to a quantum system
described by a separable, infinite-dimensional Hilbert space $\mathcal{H}$.

\begin{enumerate}
\item  When we prepare the system in a given pure quantum state, given by a
complex one-dimensional subspace, actually we prepare the system in a \emph{%
\ hidden state} described by a non-zero vector $\psi $ in the assigned
complex line (the \emph{apparent pure state}).

We suppose the existence of a measure $\eta _{\mathbb{C\cdot \psi }}$ on
every complex line $\mathbb{C\cdot \psi }$ expressing the probability to
find the apparent state $\mathbb{C\cdot \psi }$ in the hidden state $\psi $.

\item  When we prepare a measure apparatus corresponding to a precise
quantum observable, a self-adjoint operator $T$ (here supposed bounded, for
simplicity), actually we prepare a measure apparatus corresponding to a 
\emph{hidden observable} described by a function $f:\mathcal{H}\rightarrow 
\mathbb{R}$ giving the values effectively observed.

The \emph{apparent observable} $T$ and the hidden observable $f$ are joined
by the condition: 
\begin{equation*}
\left\langle E_{B}^{T}\right\rangle _{\psi }=\eta _{\mathbb{C\cdot \psi }%
}(f^{-1}(B)\cap \mathbb{C\cdot \psi )}
\end{equation*}

expressing the equality between two ways to compute the probability that the
observable will give a value falling in the borel subset $B$ of $\mathbb{R}$
for the apparent state $\mathbb{C\cdot \psi }$.

\item  When we perform a quantum test, described by a projector $E$,
actually we identify a subset $L$ of $\mathcal{H}$ (a \emph{hidden test }or
a \emph{hidden proposition}); the hidden test $L$ and the \emph{apparent test%
} $E$ are joined by the equality: 
\begin{equation*}
\left\langle E\right\rangle _{\psi }=\eta _{\mathbb{C\cdot \psi }}(L\cap 
\mathbb{C\cdot \psi )}
\end{equation*}

between two expressions of the probability that the test will receive an
affirmative answer for the apparent state $\mathbb{C\cdot \psi }$.

\item  When we prepare a mixed quantum state described by a density matrix $%
D $, actually we preapare the state in a hidden mixed state described by a
probability measure $\mu $ on $\mathcal{H}$. The \emph{apparent mixed state} 
$D$ and the hidden mixed state $\mu $ are joined by the relation: 
\begin{equation*}
Trace\left[ E\cdot D\right] =\mu (L)
\end{equation*}

for every couple of a hidden test $L$ and the corresponding apparent test $E$%
, giving the equality between two expressions of the probability that the
test will receive an affirmative answer for the apparent mixed state $D$.

\item  When an apparent observable $T$ acts on an apparent mixed state $D$
actually a hidden observable $f$ acts on a hidden mixed state $\mu $ and
these objects are joined together by the condition: 
\begin{equation*}
Trace\left[ T\cdot D\right] =\int f\cdot d\mu
\end{equation*}

expressing the equality \ between two ways to compute the expectation value
for the observable acting on the apparent mixed state $D$

\item  A hidden observable $f$ corresponding to the self-adjoint operator $T$
takes almost all its values in $spec\left[ T\right] $, that is it can take
values out of $spec\left[ T\right] $ only on a subset of $\mathcal{H}$ of
measure zero; so changing the function $f$ on a set of measure zero you can
always get all its values in $spec\left[ T\right] $.
\end{enumerate}

It is important to declare that these results depend, in our opinion,
vitally on the hypothesis that behind an apparent observable there are
several hidden observables: we don't consider possible to choose a function
for every self-adjoint operator in a reasonable way (cfr. section 3).

This places the present article strongly inside the contextual position: the
experimental values observed depend not only on the variety of hidden states
behind an apparent state but also on the existence of several hidden
observables behind the apparent observable considered, each corresponding to
a different experimental context (cfr. [K-S], Ghirardi in [B] 4.6.5, [G-D]).

In particular we don't consider the sum $f+g$ of two hidden observables $f$
and $g$ is, in general, again a hidden observable even when the observables
correspond to compatible (commuting) operators.We push toward a renforcement
and a clarification of what should be called a ''context'': observable
functions in the same context can be summed, functions in different contexts
cannot and functions corresponding to non compatible observables are never
in the same context.

This is compensated however by the possibility to find always summable
hidden observables for compatible operators and more generally by the
possibility to find an algebra of hidden functions corresponding to an
assigned commutativa algebra of operators.

In this way we avoid to fall in the hypothesis that bring to some no-go
theorem (cfr. [F1], [F2], [K-S],[J-P],[M1],[M2],[P]).

The last section proves, under reasonable hypothesis, that the theory
developped here is unique up to isomorphisms.

\bigskip

\begin{center}
{\Large List of Symbols}

\bigskip
\end{center}

$\nu _{F}=$ measure induced by the monotone function $F$ : \ $\nu
_{F}(]a,b])=F(b)-F(a)$

$\varphi _{\ast }\mu =$ image measure of the measure $\mu $ via $\varphi $: $%
\ \ (\varphi _{\ast }\mu )(A)=\mu (\varphi ^{-1}(A))$

$\widetilde{F}=$ quasi-inverse function of the monotone function $F$: \ $%
\widetilde{F}(s)=Inf\left\{ r:\text{ }F(r)\geq s\right\} $

$b(T)=b\circ T=$ function of the bounded self-adjoint operator $T$

$spec[T]=$ spectrum of the bounded self-adjoint operator $T$

$Trace[T]=$ trace of the bounded self-adjoint operator $T$

$\left\langle T\right\rangle _{\psi }=$ expectation value for the bounded
self-adjoint operator $T$ on $\psi $

$E_{B}^{T}=$ projector associated to the borel subset $B$ in the spectral
measure of $T$

$\mathcal{B}_{sa}(\mathcal{H})=\mathcal{\ }$vector space of bounded
self-adjoint operators of $\mathcal{H}$

$DM(\mathcal{H})=$ the set of all density matrices on $\mathcal{H}$

$PR(\mathcal{H)}$ $=$ the space of all orthogonal projection operators of $%
\mathcal{H}$

\bigskip

\section{\protect\Large The hidden observables}

\begin{center}
\bigskip
\end{center}

\textit{From now on we will fix on the borelian subsets of the set }$\mathbb{%
C}$ \textit{of complex numbers\ a probability measure }$\eta $\textit{\
without atoms and invariant by rotations.}

\textit{For such a measure it is always possible to find borel maps }$%
\varphi :\mathbb{C}\rightarrow \left] 0,1\right[ $\textit{\ such that }$%
\varphi _{\ast }\eta =\lambda $\textit{\ (where }$\lambda $\textit{\ denotes
the Lebesgue measure on }$\left] 0,1\right[ $\textit{\ ) and subsets with
any assigned measure between }$0$\textit{\ and }$1$\textit{\ (the space} $%
\mathbb{C}$ \textit{with the measure} $\eta $\textit{\ is a standard
nonatomic probability space})\textit{.}

\textit{Let }$(\mathcal{H},\left\langle ,\right\rangle )$\textit{\ be a
separable Hilbert space over }$\mathbb{C}$ \textit{of infinite dimension, on
every complex line }$\mathbb{C}\cdot \psi $\textit{\ (with \ }$\left\| \psi
\right\| =1$\textit{) there is just one probability measure }$\eta _{\mathbb{%
C\cdot \psi }}$\textit{\ such that }$\eta _{\mathbb{C\cdot \psi }}(B\cdot
\psi )=\eta (B)$\textit{\ for every borel subset }$B$\textit{\ of }$\mathbb{C%
}$\textit{. On }$\mathcal{H}$\textit{\ we will consider the }$\sigma $%
\textit{-algebra of subsets }(\textit{called} \emph{pseudo-borel subsets}) $%
A $\textit{\ such that for every complex line }$\mathbb{C}\cdot \psi $%
\textit{\ the intersection }$A\cap \mathbb{C}\cdot \psi $\textit{\ is a
borel subset of }$\mathbb{C}\cdot \psi $\textit{. Correspondingly a map }$f:%
\mathcal{H}$\textit{\ }$\rightarrow \mathbb{R}$\textit{\ will be called} 
\emph{a pseudo-borel function}\textit{\ if }$f^{-1}(B)$\textit{\ is a
pseudo-borel subset of }$\mathcal{H}\ $\textit{for every borel subset }$B$%
\textit{\ of }$\mathbb{R}$\textit{. A pseudo-borel subset }$A$\textit{\ of }$%
\mathcal{H}$\textit{\ will be called} \emph{a zero measure subset }\textit{%
if every intersection }$A\cap \mathbb{C}\cdot \psi $\textit{\ has measure
zero in }$\mathbb{C}\cdot \psi $\textit{.}

\textit{We will use in the following} $\mathcal{H\setminus }\left\{
0\right\} $ \textit{as the total space of hidden states and each} $\mathbb{C}%
\cdot \psi \mathcal{\setminus }\left\{ 0\right\} $ \textit{as the set of
hidden states behind each quantum state} $[\psi ]$ \textit{of the complex
projective space} $\mathbb{P}_{\mathbb{C}}\mathbb{(\mathcal{H})}$; \textit{%
the constant addiction we will make of the element} $0$ \textit{to these
sets should not create confusion and it is made only with the hope to
simplify the notations (if you prefer you can simply forget everywhere the
element} $0$).

\textit{You can easily check that all the theory developed in this article
works equally well if you consider as total hidden space a generic set} $%
\Lambda $ \textit{instead of} $\mathcal{H\setminus }\left\{ 0\right\} $ 
\textit{and a partition of} $\Lambda $ \textit{in a family of subsets} $%
\Lambda _{\lbrack \psi ]}$ (\textit{where} $[\psi ]$ \textit{varies in} $%
\mathbb{P}_{\mathbb{C}}\mathbb{(\mathcal{H})}$) \textit{instead of the
partition} \textit{into the lines} $\mathbb{C}\cdot \psi \mathcal{\setminus }%
\left\{ 0\right\} $ \textit{of }$\mathcal{H\setminus }\left\{ 0\right\} $, 
\textit{each }$\Lambda _{\lbrack \psi ]}$ \textit{furnished with} \textit{a }%
$\sigma $-\textit{algebra of subsets and a} \textit{probability measure }$%
\eta _{\lbrack \psi ]}$ \textit{making} $\Lambda _{\lbrack \psi ]}$ a 
\textit{standard nonatomic probability space.}

\bigskip

\begin{definition}
An essentially bounded pseudo-borel function $f:\mathcal{H\rightarrow }%
\mathbb{R}$ will be called \emph{a function with orthodox mean values } if
there exists a (unique) bounded self-adjoint operator $T$ such that : $\int_{%
\mathbb{C\cdot \psi }}f\cdot d\eta _{\mathbb{C\cdot \psi }}=\left\langle
T\right\rangle _{\psi }$ for every $\psi \in \mathcal{H\setminus }\left\{
0\right\} $.
\end{definition}

\textit{The set }$\mathcal{F}$\textit{\ of all functions with orthodox mean
values is a real vector space; the self-adjoint operator associated to a
function by the previous definition is uniquely determined. The map }$\sigma
:\mathcal{F}\rightarrow \mathcal{B}_{sa}(\mathcal{H})$\textit{\ so defined
is a real linear map.}

\textit{We will use the symbol }$\mathcal{B}$\textit{\ to denote the algebra
of all real borel functions } \textit{sending bounded subsets of }$\mathbb{R}
$\textit{\ in bounded subsets of }$\mathbb{R}$\textit{. This algebra
contains the constant functions, all the polynomials, all the continuous
functions and is closed by compositions.}

\textit{Note that for every bounded self-adjoint operator }$T$ \textit{all
the operators} $b\circ T$ \textit{(with }$b$\textit{\ in }$\mathcal{B}$%
\textit{) are well-defined bounded self-adjoint operators.}

\begin{definition}
A function $f$ with orthodox mean values will be called a \emph{(essentially
bounded) hidden observable function on $\mathcal{H}$ } if for every function 
$b\ $in $\mathcal{B}$ the composition $b\circ f$ \ has orthodox mean values\
and $\sigma (b\circ f)=b\circ \sigma (f).$
\end{definition}

\textit{We will denote by }$\mathcal{O}$\textit{\ the set of all hidden
observable functions on }$\mathcal{H}$\textit{. Given a hidden observable }$%
\ f$\textit{\ all functions }$b\circ f$\textit{\ (with }$b$\textit{\ in }$%
\mathcal{B}$\textit{) are hidden observable functions. A function }$g$ 
\textit{differing only on a zero-measure subset from an observable function }%
$f$ \textit{\ is also an observable function with }$\sigma (f)=\sigma (g)$%
\textit{.}

\textit{The sum or the product of hidden observable functions is not, in
general, a hidden observable function.}

\textit{If} $f$ is \textit{an observable function with }$\sigma (f)=T$ 
\textit{then} $E_{B}^{T}=\chi _{B}\circ T=\sigma \left( \chi
_{f^{-1}(B)}\right) $ \textit{for every borel subset} $B$ of $\mathbb{R}$ 
\textit{and } $\left\langle E_{\left( -\infty ,s\right] }^{T}\right\rangle
_{\psi }=\eta _{\mathbb{C\cdot \psi }}(f^{-1}\left( -\infty ,s\right] \cap 
\mathbb{C\cdot \psi })$ \textit{for every} $s$ \textit{in} $\mathbb{R}$. 
\textit{That is the borel measure} $\nu _{F_{\psi }}$ \textit{induced by the
function} $F_{\psi }(s)=\left\langle E_{\left( -\infty ,s\right]
}^{T}\right\rangle _{\psi }$ \textit{coincides with the image measure} $%
(f\mid _{\mathbb{C\cdot \psi }})_{\ast }\eta _{\mathbb{C\cdot \psi }}$.

\begin{theorem}
\textit{\ An essentially bounded pseudo-borel function} $f:\mathcal{%
H\rightarrow }\mathbb{R}$ \textit{is a hidden observable if and only if
there exists a (unique) bounded self-adjoint operator} $T$ \textit{such that}
: $\int_{\mathbb{C\cdot \psi }}f^{n}\cdot d\eta _{\mathbb{C\cdot \psi }%
}=\left\langle T^{n}\right\rangle _{\psi }$ \textit{for every} $\psi \in 
\mathcal{H\setminus }\left\{ 0\right\} $ \textit{and every} $n\geq 0$.
\end{theorem}

\begin{proof}
Obviously you have the equality: $\int_{\mathbb{C\cdot \psi }}b\circ f\cdot
d\eta _{\mathbb{C\cdot \psi }}=\left\langle b\circ T\right\rangle _{\psi }$
for every polynomial function $b$; to prove the same equality for a generic
continuous function $b$ consider a sequence $\left\{ b_{n}\right\} _{n\geq
1} $ of polynomials uniformly converging to $b$ on a closed interval $[-N,N]$
containing $spec\left[ T\right] $; standard converging properties for the
integrals and for the operators imply the desired equality.

Finally to prove the equality for a generic function $b$ in $\mathcal{B}$
consider a sequence $\left\{ b_{n}\right\} _{n\geq 1}$ of continuous
functions converging in the $L^{1}([-N,N])$ norm to $b$.
\end{proof}

\begin{definition}
A pseudo-borel subset $L$ of $\mathcal{H}$ will be called a \emph{hidden
proposition} if its characteristic function is a hidden observable.
\end{definition}

\textit{We will denote by }$\mathcal{L}$\textit{\ the set of all hidden
propositions of }$\mathcal{H}$\textit{. The set }$\mathcal{L}$\textit{\ is
called the} \emph{hidden logic} \textit{of }$\mathcal{H}$.

\textit{The empty set and }$\mathcal{H}$\textit{\ are hidden propositions;
the complement of a hidden proposition is again a hidden proposition. Every
pseudo-borel zero-measure subset }$L$ \textit{of }$\mathcal{H}$ \textit{is a
hidden proposition with} $\sigma (\chi _{L})=0$.

\textit{The union or the intersection of two hidden propositions is not, in
general, a hidden proposition.}

\begin{theorem}
Let $L$ be a pseudo-borel subset of $\mathcal{H}$ with $\chi _{L}$ in $%
\mathcal{F}$, the subset $L$ is a hidden proposition if and only if \textit{%
\ the operator }$\sigma (\chi _{L})$\textit{\ is a projector of }$\mathcal{H}
$.
\end{theorem}

\begin{proof}
$(\Longrightarrow )$If $\chi _{L}$ is an observable then $\sigma (\chi
_{L})^{2}=$ $\sigma (\chi _{L}^{2})=\sigma (\chi _{L})$

$(\Longleftarrow )$Let $\sigma (\chi _{L})=E$ be a projector. Whatever is $b$
in $\mathcal{B}$ we have: $b\circ \chi _{L}=[b(1)-b(0)]\cdot \chi
_{L}+b(0).1 $, then taken $c(x)=[b(1)-b(0)]\cdot x+b(0)$ the function $c$ is
in $\mathcal{B}$ and we have $\sigma (b\circ \chi _{L})=[b(1)-b(0)]\cdot
E+b(0)\cdot I=c\circ E=b\circ E$ since $b$ and $c$ take the same values on
the spectrum of $E$ (cfr. [W] ex. 7.36 pag. 210).
\end{proof}

\textit{If} $L$ \textit{and} $M$ \textit{are hidden propositions then} $%
L\subset M$ \textit{implies} $\sigma (\chi _{L})\leq \sigma (\chi _{M})$; $%
\sigma (\chi _{\complement L})=I-\sigma (\chi _{L})$; $L\cap M=\emptyset $ 
\textit{implies} $\sigma (\chi _{L})\cdot \sigma (\chi _{M})=0$. \textit{If} 
$\left\{ L_{n}\right\} _{n\geq 1}$ \textit{is a family of disjoint hidden
propositions then} $\bigcup_{n\geq 1}L_{n}$ \textit{is a hidden proposition
with} $\sigma (\chi _{\bigcup_{n\geq 1}L_{n}})=\sum_{n\geq 1}$ $\sigma (\chi
_{L_{n}})$.

\begin{theorem}
A function $f$ with orthodox mean values is a hidden observable if and only
if for every borel subset $B$ of $\mathbb{R}$ the subset $f^{-1}(B)$ is a
hidden proposition.
\end{theorem}

\begin{proof}
$(\Longrightarrow )\chi _{f^{-1}(B)}=\chi _{B}\circ f$ is in $\mathcal{F}$
for every borel subset\textit{\ }$B$\textit{\ }of\textit{\ }$\mathbb{R}$ and 
$\sigma (\chi _{f^{-1}(B)})^{2}=$ $\sigma (\chi _{B}^{2}\circ f)=$ $\sigma
(\chi _{B}\circ f)$ is a projector.

$(\Longleftarrow )$Since the family $\left\{ f^{-1}\left( -\infty ,s\right]
\right\} _{s\in \mathbb{R}}$is a family of hidden propositions essentially
empty for $s$ small and essentially $\mathcal{H}$ for $s$ big, the family $%
\left\{ \sigma (\chi _{f^{-1}\left( -\infty ,s\right] })\right\} _{s\in 
\mathbb{R}}$is the spectral family of a bounded self-adjoint operator $T$ .

Therefore: $\left\langle E_{\left( -\infty ,s\right] }^{T}\right\rangle
_{\psi }=\eta _{\mathbb{C\cdot \psi }}(f^{-1}\left( -\infty ,s\right] \cap 
\mathbb{C\cdot \psi })$ for every $\ s$ in $\mathbb{R}$ and every $\psi \neq
0$. In other words the borel measure $\nu _{F_{\psi }}$ induced by the
function $F_{\psi }(s)=\left\langle E_{\left( -\infty ,s\right]
}^{T}\right\rangle _{\psi }$ coincides with the image measure $(f\mid _{%
\mathbb{C\cdot \psi }})_{\ast }\eta _{\mathbb{C\cdot \psi }}$. So we can
compute: $\int_{\mathbb{C\cdot \psi }}b\circ f\cdot d\eta _{\mathbb{C\cdot
\psi }}=\int_{\mathbb{R}}b\cdot d\left[ (f\mid _{\mathbb{C\cdot \psi }%
})_{\ast }\eta _{\mathbb{C\cdot \psi }}\right] =\int_{\mathbb{R}}b\cdot d\nu
_{F_{\psi }}=\left\langle b\circ T\right\rangle _{\psi }$whatever is $b$ in $%
\mathcal{B}$ and we can state that all the functions $b\circ f$ have
orthodox mean values with $\sigma (b\circ f)=b\circ T=b\circ \sigma (f)$.
\end{proof}

\begin{theorem}
For every self-adjoint bounded operator $T$ it is possible to find a hidden
observable $f$ such that $\sigma (f)=T$.
\end{theorem}

\begin{proof}
Let spec$\left[ T\right] \subset \left[ -A,A\right] $, for every $\psi \neq
0 $ the monotone function $F_{\psi }(s)=\left\langle E_{\left( -\infty ,s%
\right] }^{T}\right\rangle _{\psi }$is $0$ before $-A$ and $1$ after $+A$,
therefore its quasi-inverse $\widetilde{F_{\psi }}$ is absolutely bounded by 
$A$ and has the property: $(\widetilde{F_{\psi }})_{\ast }\lambda _{\left]
0,1\right[ }=\nu _{F_{\psi }}$(cfr. \ [K and S]\ \ thm. 4 p. 94)

Let's fix for every complex line $\mathbb{C\cdot \psi }$ in $\mathcal{H}$ a
borel map $\gamma \mid _{\mathbb{C}^{\ast }\mathbb{\cdot \psi }}:(\mathbb{%
C\setminus }\left\{ 0\right\} )\cdot \mathbb{\psi \rightarrow }\left] 0,1%
\right[ $ such that $(\gamma \mid _{\mathbb{C}^{\ast }\mathbb{\cdot \psi }%
})_{\ast }\eta _{\mathbb{C\cdot \psi }}=\lambda _{\left] 0,1\right[ }$.
Therefore the function $f:\mathcal{H\rightarrow }\mathbb{R}$ defined by $%
f(\psi )=\left( \widetilde{F_{\psi }}\circ \left( \gamma \mid _{\mathbb{C}%
^{\ast }\mathbb{\cdot \psi }}\right) \right) (\psi )$ when $\psi $ is in $(%
\mathbb{C\setminus }\left\{ 0\right\} )\cdot \mathbb{\psi }$ and defined $0$
in the vector $0$ is absolutely bounded by $A$ and it verifies: $(f\mid _{%
\mathbb{C\cdot \psi }})_{\ast }\eta _{\mathbb{C\cdot \psi }}=\nu _{F_{\psi
}} $ for every line $\mathbb{C\cdot \psi }$. Proceeding as in the previous
proof this implies that all the functions $b\circ f$ have orthodox mean
values and moreover $\sigma (b\circ f)=b\circ T=b\circ \sigma (f)$. That is $%
f$ is a hidden observable and $\sigma (f)=T$.
\end{proof}

\textit{Remembering the definition of a quasi-inverse function, the
observable }$f$ \textit{defined in the previous proof is given explicitally
by the expression:} 
\begin{equation*}
f_{\gamma }(\psi )=\min \left\{ r\in \mathbb{R:}\left\langle E_{\left(
-\infty ,r\right] }^{T}\right\rangle _{\psi }\geq \gamma \left( \psi \right)
\right\}
\end{equation*}

\textit{We could prove that this expression is pratically exhaustive:
infact, assigned the operator }$T$ \textit{and a function }$f$\textit{\ such
that }$\sigma (f)=T$\textit{, it is possible to find a map }$\gamma $\textit{%
\ such that }$f=f_{\gamma }$ \textit{(up to a zero measure set) . However we
will not present here the proof of this theorem since we will not need this
property in the following.}

\begin{theorem}
Given a self-adjoint bounded operator $T$ a hidden observable $f$ such that $%
\sigma (f)=T$ modified on a set of measure zero verifies $\overline{f(%
\mathcal{H})}=spec\left[ T\right] $.
\end{theorem}

\begin{proof}
There exists a biggest open subset $W$ of $\mathbb{R}$ such that $f^{-1}(W)$
is a zero measure subset of $\mathcal{H}$. Since $f(\mathcal{H)\setminus }W$
is not empty we can redefine the function $f$ on $f^{-1}(W)$ with a value
chosen in $f(\mathcal{H)\setminus }W$; this new function $f$ is again an
observable with $\sigma (f)=T$ and moreover does not take values in $W$,
that is does not allow non-empty open subsets $U$ of $\mathbb{R}$ with $%
f^{-1}(U)$ of zero measure. A value $y$ of $\mathbb{R}$ is not in $spec\left[
T\right] $ if and only if $E_{\left] y-\varepsilon ,y+\varepsilon \right[
}^{T}=0$ for a suitable $\varepsilon >0$ that if and only if $f^{-1}\left]
y-\varepsilon ,y+\varepsilon \right[ $ is a zero measure subset: but, for
this new function $f$ , this is equivalent to $f^{-1}\left] y-\varepsilon
,y+\varepsilon \right[ =\emptyset $ and to $y\notin \overline{f(\mathcal{H})}
$.
\end{proof}

\begin{corollary}
For every projector $E$ there exists a proposition $L$ with $\sigma (\chi
_{L})=E$.
\end{corollary}

\begin{proof}
Let $f$ be an observable such that $\sigma (f)=E$ with $f(\mathcal{H)\subset 
}spec\left[ E\right] =\left\{ 0,1\right\} $, the function $f$ is the
characteristic function of the proposition $L=f^{-1}(\left\{ 1\right\} )$.
\end{proof}

\textit{Note that for every hidden observable function} $f$ \textit{the set} 
$\mathcal{H\setminus }f^{-1}(spec\left[ \sigma (f)\right] $ \textit{is a
zero measure subset of} $\mathcal{H}$.

\section{\protect\bigskip {\protect\Large Algebras and contexts}}

\bigskip

\textit{Let's imagine to be able to build an apparatus suitable to measure
one or several quantities of the hidden system in a deterministic way (that
is you get for a given ''observable'' on a given ''hidden state'' always the
same value); this defines a precise experimental context and a family }$%
\mathcal{C}$ \textit{of all possible ''observable functions'' on} \textit{%
the total space of hidden states associated to that given experimental
context.}

\textit{In these hypothesis if you can measure} $f(\psi )$ \textit{and} $%
g(\psi )$ \textit{you can also compute} $f(\psi )+g(\psi )$\textit{, } $%
f(\psi ).g(\psi )$ \textit{and }$k\cdot f(\psi )$ \textit{for every constant 
}$k$, \textit{so }$\mathcal{C}$ \textit{must be a commutative algebra of
functions. Moreover nothing can prevent you to compute} $b(f(\psi ))$ 
\textit{where} $b$ \textit{is any available real function, therefore it is
not rescrictive to suppose} $\mathcal{C}$ \textit{also closed with respect
to the composition with the functions} $b$ \textit{of} $\mathcal{B}$. 
\textit{\ So this kind of algebra is, at least from a mathematical
viewpoint, representative of the choice of an} \textit{experimental context. 
}

\begin{example}
Let $\left\{ L_{n}\right\} _{n\geq 1}$ be a family of (non-zero measure)
pairwise disjoint hidden propositions of $\mathcal{H}$. Let's define: 
\begin{equation*}
\mathcal{C=}\left\{ f:\mathcal{H\rightarrow }\mathbb{R};\text{ \ }%
f=\sum_{n\geq 1}c_{n}\cdot \chi _{L_{n}}\text{ with }\left\{ c_{n}\right\}
_{n\geq 1}\text{ bounded}\right\}
\end{equation*}
the family $\mathcal{C}$ is an algebra of hidden observable functions closed
by the compositions with the borel functions $b$ in $\mathcal{B}$ . The
image of $\mathcal{C}$ via $\sigma $ is the commutative algebra of bounded
self-adjoint operators: 
\begin{equation*}
\mathcal{A}=\left\{ T;\text{ \ }T=\sum_{n\geq 1}c_{n}\cdot \sigma (\chi
_{L_{n}})\text{ with }\left\{ c_{n}\right\} _{n\geq 1}\text{ bounded}\right\}
\end{equation*}
and $\sigma |:\mathcal{C\rightarrow A}$ is an isomorphism of algebras.
\end{example}

\bigskip

\textit{The following theorem puts a strong limit to the existence of such
mathematical objects }$\mathcal{C}$\textit{.}

\bigskip

\begin{theorem}
Let $\mathcal{C}$ be an algebra of functions on $\mathcal{H}$ closed with
respect to the composition with the functions of $\mathcal{B}$, the
following alternative holds:

\begin{itemize}
\item  $\mathcal{C}$ is not contained in $\mathcal{O}$

or

\item  $\sigma (\mathcal{C})$ is a commutative family of bounded
self-adjoint operators.
\end{itemize}
\end{theorem}

\begin{proof}
Let's suppose $\mathcal{C\subset O}$; let's take two functions $f$ and $g$
in $\mathcal{C}$ and two borel subsets $A$ and $B$ of $\mathbb{R}$. Since $f$
and $g$ are hidden observables the subsets $L=f^{-1}(A)$ and $M=g^{-1}(B)$
are two hidden propositions in $\mathcal{H}$ with $\sigma (\chi _{L})=\sigma
(\chi _{A}\circ f)=\chi _{A}\circ \sigma (f)=E_{A}^{\sigma (f)}$ and $\sigma
(\chi _{M})=E_{B}^{\sigma (g)}$. Moreover, since $\mathcal{C}$ is closed
with respect to the composition with the functions of $\mathcal{B}$, the
functions $\chi _{L}=\chi _{A}\circ f$ and $\chi _{M}=\chi _{B}\circ g$ are
in $\mathcal{C}$.

Let $h:\mathcal{H\rightarrow }\left\{ 1,2,3,4\right\} $ be the function
taking value $1$ on the elements of $L\setminus M$, value $2$ on the
elements of $M\setminus L$, value $3$ on the elements of $L\cap M$ and value 
$4$ on the set $\complement L\cap \complement M$. Since $\chi
_{h^{-1}(1)}=\chi _{L}-\chi _{L}\cdot \chi _{M}$, $\chi _{h^{-1}(2)}=\chi
_{M}-\chi _{L}\cdot \chi _{M}$, $\chi _{h^{-1}(3)}=\chi _{L}\cdot \chi _{M}$%
, $\chi _{h^{-1}(4)}=1-\chi _{L}-\chi _{M}+\chi _{L}\cdot \chi _{M}$ and $%
\chi _{h^{-1}\left\{ n_{1},...,n_{r}\right\} }=\sum \chi _{h^{-1}(n_{i})%
\text{ }}$(when $n_{1},...,n_{r}$ are distinct numbers in $\left\{
1,2,3,4\right\} $), all the possible characteristic functions $\chi
_{h^{-1}\left\{ n_{1},...,n_{r}\right\} }$ are in the algebra $\mathcal{C}$
and by hypothesis in $\mathcal{O}$. Therefore $h$ is a hidden observable
with $L=h^{-1}(\left\{ 1,3\right\} )$ and $M=h^{-1}(\left\{ 2,3\right\} )$;
then $E_{A}^{\sigma (f)}=\sigma (\chi _{L})=E_{\left\{ 1,3\right\} }^{\sigma
(h)}$ and $E_{B}^{\sigma (g)}=\sigma (\chi _{M})=E_{\left\{ 2,3\right\}
}^{\sigma (h)}$ must commute as projectors in the same spectral measure. For
the arbitrarity of $A$ and $B$ the operators $\sigma (f)$ and $\sigma (g)$
commute.
\end{proof}

\bigskip

\textit{The theorem just proved is a kind of no-go theorem since it claims
that you cannot hope to find a hidden variable theory where you can realize
an algebra of hidden observable functions associated with an ''experimental
context'' and representing non-commuting bounded self-adjoint operators.}

\textit{Alternatively you could say that such an ''experimental context''
can be imagined but with some of its associated functions out of} $\mathcal{O%
}$\textit{, that is functions whose mean values have a non-ortodox behaviour
(precisely functions} $f$ \textit{such that whatever is} $T$ \textit{bounded
self-adjoint operator there exists a non-zero vector} $\psi $ \textit{and a
non-negative integer} $n$ \textit{with} $\int_{\mathbb{C\cdot \psi }%
}f^{n}\cdot d\eta _{\mathbb{C\cdot \psi }}\neq \left\langle
T^{n}\right\rangle _{\psi }$).

\bigskip

\textit{If you don't ask to represent} \textit{non-commuting bounded
self-adjoint operators you have a positive answer:}

\begin{theorem}
Let $\mathcal{A}$ be a commutative algebra of bounded self-adjoint operators
closed by the compositions with the functions of $\mathcal{B}$ , it is
possible to find an algebra $\mathcal{C}$ of hidden observable functions
closed by the compositions with the functions of $\mathcal{B}$ such that:

\begin{enumerate}
\item  $\sigma (\mathcal{C})=\mathcal{A}$

\item  $\sigma \mid :\mathcal{C}\rightarrow \mathcal{A}$ is an algebra
homomorphism

\item  $\sigma (h)=0$ for $h$ in $\mathcal{C}$ if and only if $h$ is zero
out of a pseudo-borel null set.
\end{enumerate}
\end{theorem}

\begin{proof}
Let $\mathcal{A=}\left\{ A_{i}\right\} _{i\in I}$, since $\mathcal{H}$ is
separable and all the operators $A_{i}$ commute each other there exists a
self-adjoint operator $T_{0}$ (may be not bounded) and a family of borel
functions $\left\{ b_{i}\right\} _{i\in I}$ such that $A_{i}$ =$b_{i}\circ
T_{0}$ for every $i$ in $I$ (cfr. [VN] and [V]). Since all the operators $%
A_{i}$ $\ $are bounded it is possible to correct the borel functions $%
\left\{ b_{i}\right\} _{i\in I}$ and take them all bounded. Let $\mathcal{A}%
_{1}=\mathcal{B\circ }T_{0}$ this is a commutative algebra of self-adjoint
operators containing $\mathcal{A}$.

Proceedings as in the proof of theorem 4 it is possible to find a pseudo
borel function $f_{0}:\mathcal{H\rightarrow }\mathbb{R}$ with $(f_{0}\mid _{%
\mathbb{C\cdot \psi }})_{\ast }\eta _{\mathbb{C\cdot \psi }}=\nu _{F_{\psi
}^{T}}$ for every complex line $\mathbb{C\cdot \psi }$ (where $F_{\psi
}^{T_{0}}(s)=\left\langle E_{\left( -\infty ,s\right] }^{T_{0}}\right\rangle
_{\psi }$ and $\nu _{F_{\psi }^{T_{0}}}(B)=\left\langle
E_{B}^{T_{0}}\right\rangle _{\psi }$ for every borel subset $B$ in $\mathbb{R%
}$).

We can prove that the function $b\circ f_{0}$ is a hidden observable for
every $b$ in $\mathcal{B}$; in fact $b_{\ast }f_{0\ast }\eta _{\mathbb{%
C\cdot \psi }}(B)=\left\langle E_{b^{-1}B}^{T_{0}}\right\rangle _{\psi
}=\left\langle E_{B}^{b\circ T_{0}}\right\rangle _{\psi }$ for every borel
subset $B$ in $\mathbb{R}$, that is $b_{\ast }f_{0\ast }\eta _{\mathbb{%
C\cdot \psi }}=\nu _{F_{\psi }^{b\circ T_{0}}}$ therefore: $\int_{\mathbb{%
C\cdot \psi }}b\circ f_{0}\cdot d\eta _{\mathbb{C\cdot \psi }}=\int_{\mathbb{%
R}}id\cdot d(b_{\ast }f_{0\ast }\eta _{\mathbb{C\cdot \psi }})=\int_{\mathbb{%
R}}id\cdot d\nu _{F_{\psi }^{b\circ T}}=$ $\left\langle b\circ
T_{0}\right\rangle _{\psi }$. This proves that every $b\circ f_{0}$ $\ $has
orthodox mean values and also that $\sigma (b\circ f_{0})=b\circ T_{0}$ for
every $b$ in $\mathcal{B}$; this proves that every $b\circ f_{0}$ $\ $is a
hidden observable.

Let $\mathcal{C}_{1}=\mathcal{B\circ }$ $f_{0}$ this is an algebra of hidden
observable funtions closed by the compositions with the functions of $%
\mathcal{B}$, the points 1. and 2. follow immediately for $\mathcal{A}_{1}$
and $\mathcal{C}_{1}$. If $\sigma (b\circ f_{0})=0$ then $\eta _{\mathbb{%
C\cdot \psi }}((b\circ f_{0})^{-1}(\mathbb{R\setminus }\left\{ 0\right\}
)\cap \mathbb{C\cdot \psi )=}0$ for every complex line $\mathbb{C\cdot \psi }
$ and $(b\circ f_{0})^{-1}(\mathbb{R\setminus }\left\{ 0\right\} )$ is a
pseudo-borel null subset of $\mathcal{H}$. This proves also the point 3.

To conclude the proof let's consider the algebra $\mathcal{C}=(\sigma \mid
)^{-1}(\mathcal{A)}$, if $f=b\circ f_{0}$ is in $\mathcal{C}$ and $c$ is in $%
\mathcal{B}$ then $c\circ f$ is in $\mathcal{C}_{1}$ with $\sigma (c\circ
f)=c\circ \sigma (f)$ in $\mathcal{A}$ and therefore $c\circ f$ is in $%
\mathcal{C}$.
\end{proof}

\textit{In the example given above the inverse map} $\Phi =\left( \sigma
|\right) ^{-1}:\mathcal{A\rightarrow C\subset O}$ \textit{is a map with the
property} $\Phi (b\circ T)=b\circ \Phi (T)$ \textit{for every} $b$ \textit{in%
} $\mathcal{B}$ \textit{and every }$T$ \textit{in} $\mathcal{A}$ \textit{and
express the possibility, at least from a mathematical viewpoint, to fix an
''experimental context'' where to measure deterministic observables
corresponding one-to-one to the observables in} $\mathcal{A}$.

\textit{\ A map of this kind extended as much as possible} \textit{would
give a general family of ''experimental contexts'' pacifically coexisting.
But it is not possible to go too far in this direction:}

\begin{theorem}
Let $\mathcal{D}$ be a subset of $\mathcal{B}_{sa}(\mathcal{H})$ closed by
the compositions with the functions of $\mathcal{B}\ $and $\Phi :\mathcal{D}%
\rightarrow \mathcal{O}$ \ a map such that: $\Phi (b\circ T)=b\circ \Phi (T)$
for every $b$ in $\mathcal{B}$ and every $T$ in $\mathcal{D}$, then $%
\mathcal{D}$ is properly contained in $\mathcal{B}_{sa}(\mathcal{H})$.
\end{theorem}

\begin{proof}
Let's suppose, conversely, that such a map $\Phi :\mathcal{B}_{sa}(\mathcal{H%
})\rightarrow $ $\mathcal{O}$ exists. For every projector $E$ let $%
f_{E}=\Phi (E)$, since $f_{E}^{2}=\Phi (E)^{2}=\Phi (E^{2})=\Phi (E)=$ $%
f_{E} $ the function $f_{E}$ is a characteristic function with $f_{E}=\chi
_{L_{E}} $ where $L_{E}$ is the hidden proposition$\ L_{E}=f_{E}^{-1}(1)$.

In particular $\Phi (I)=\chi _{L}$; if $L\subsetneqq \mathcal{H}$ let's take 
$\varphi \notin L$ and a function $b$ in $\mathcal{B}$ with $b(0)\neq 0$. We
have $\Phi (b\circ I)(\varphi )=\Phi (b(1)\cdot I)(\varphi )=\Phi \left[
(b(1)\cdot id_{\mathbb{R}})\circ I\right] (\varphi )=\left[ (b(1)\cdot id_{%
\mathbb{R}})\circ \chi _{L}\right] (\varphi )=0$ but $(b\circ \chi
_{L})(\varphi )=b(0)\neq 0$.

Therefore $L=\mathcal{H}$ and $\Phi (I)=1$.

Let's fix once for all a unit vector $\psi _{0}$ and let's define the map $G:%
\mathbb{S}(1)\rightarrow \left\{ 0,1\right\} $ given by $G(\psi )=\chi
_{L_{E[\psi ]}}(\psi _{0})$ where $E[\psi ]$ is the orthogonal projector on
the line $\mathbb{C\cdot \psi }$.

Let's consider an orthonormal base $\left\{ \psi _{n}\right\} _{n\geq 1}$ of 
$\mathcal{H}$; if $\left\{ a_{n}\right\} _{n\geq 1}$ is a bounded injective
sequence of real numbers the operator $T=\sum_{n\geq 1}a_{n}\cdot E[\psi
_{n}]$ is a bounded self-adjoint operator with $E[\psi _{n}]=E_{\left\{
a_{n}\right\} }^{T}$ $=\chi _{\left\{ a_{n}\right\} }\circ T$ for every $%
n\geq 1$. Therefore taken $f=\Phi (T)$ we get: $\chi _{f^{-1}(\left\{
a_{n}\right\} )}=\chi _{\left\{ a_{n}\right\} }\circ f=\chi _{\left\{
a_{n}\right\} }\circ \Phi (T)=\Phi (\chi _{\left\{ a_{n}\right\} }\circ
T)=\Phi (E_{\left\{ a_{n}\right\} }^{T})=\Phi (E[\psi _{n}])=\chi
_{L_{E[\psi _{n}]}}$ that \ is $L_{E[\psi _{n}]}=f^{-1}(\left\{
a_{n}\right\} )$.

Since $I=\sum_{n\geq 1}E[\psi _{n}]=\sum_{n\geq 1}E_{\left\{ a_{n}\right\}
}^{T}=E_{\cup \left\{ a_{n}\right\} }^{T}=\chi _{\cup \left\{ a_{n}\right\}
}\circ T$ we have $1=\Phi (I)=\chi _{\cup \left\{ a_{n}\right\} }\circ
f=\sum \chi _{f^{-1}(\left\{ a_{n}\right\} )}$ that is $\left\{ L_{E[\psi
_{n}]}\right\} $ is a partition of $\mathcal{H}$ and the vector $\psi _{0}$
lies exactly in one of the sets $\left\{ L_{E[\psi _{n}]}\right\} $.

Therefore $\sum_{n\geq 1}G(\psi _{n})=1$ for every orthonormal base $\left\{
\psi _{n}\right\} _{n\geq 1}$, then $G$ is, by definition (cfr. [G]), a
Gleason frame function of weight $1$. Since $Dim(\mathcal{H)\geq }4$ there
exists a bounded self-adjoint operator $S$ such that $\left\langle
S\right\rangle _{\psi }=G(\psi )$ for every $\psi $ in $\mathbb{S}(1)$; the
continuity of $\left\langle S\right\rangle $ implies $S=0$ or $S=I$ and then 
$G=0$ or $G=1.$ In both cases we don't have $\sum_{n\geq 1}G(\psi _{n})=1$
for an orthonormal base $\left\{ \psi _{n}\right\} _{n\geq 1}$:
contradiction.
\end{proof}

\bigskip

\section{\protect\Large The hidden mixed states}

\begin{center}
\bigskip
\end{center}

\begin{definition}
A probability measure $\mu $ defined on the pseudo-borel subsets of $%
\mathcal{H}$ will be called a \emph{hidden mixed state} on $\mathcal{H}$ if
for every couple of hidden propositions $L$ and $M$ we have $\mu (L)=\mu (M)$
when $\eta _{\mathbb{C\cdot \psi }}(L\cap \mathbb{C\cdot \psi })=\eta _{%
\mathbb{C\cdot \psi }}(M\cap \mathbb{C\cdot \psi })$ for every complex line $%
\mathbb{C\cdot \psi }$.
\end{definition}

A probability measure $\mu $ defined on the pseudo-borel subsets of $%
\mathcal{H}$ is a hidden mixed state if $\sigma (\chi _{L})=\sigma (\chi
_{M})$ implies $\mu (L)=\mu (M)$. If $\mu $ is a hidden mixed state and $\mu
^{\prime }$ is a measure taking the same values of $\mu $ on the hidden
propositions then also $\mu ^{\prime }$ is also a hidden mixed state.

The measure defined by $\mu (A)=\eta _{\mathbb{C\cdot \psi }}(A\cap \mathbb{%
C\cdot \psi })$ on the pseudoborel subsets of $\mathcal{H}$ (where $\psi
\neq 0$) is a hidden mixed state. If $\left\{ \mu _{k}\right\} _{k\geq 1}$
is a sequence of hidden mixed states and $\left\{ w_{k}\right\} _{k\geq 1}$
is a sequence of real numbers in $\left[ 0,1\right] $ with $\sum_{k}w_{k}=1$
then $\sum_{k}w_{k}\cdot \mu _{k}$ is a hidden mixed state.

\begin{lemma}
Let $D=\sum_{k}w_{k}\cdot E_{\mathbb{C\cdot \psi }_{k}}$ be a density matrix
(where $\left\{ \psi _{k}\right\} _{k\geq 1}$ is an orthonormal basis of $%
\mathcal{H}$ and $\left\{ w_{k}\right\} _{k\geq 1}$ is a sequence of real
numbers in $\left[ 0,1\right] $ with $\sum_{k}w_{k}=1$) for every bounded
self-adjoint operator $T$ we have: 
\begin{equation*}
Trace[T\cdot D]=\sum_{k}w_{k}\cdot \left\langle T\right\rangle _{\mathbb{%
\psi }_{k}}
\end{equation*}
\end{lemma}

\begin{proof}
It is enough to compute the trace using the orthonormal basis $\left\{ \psi
_{k}\right\} _{k\geq 1}$.
\end{proof}

\begin{theorem}
Let $\mu $ be a probability measure on $\mathcal{H}$, the measure $\mu $ is
a hidden mixed state on $\mathcal{H}$ if and only if there exists exactly
one density matrix $D$ on $\mathcal{H}$ such that: $\mu (L)=Trace\left[
\sigma (\chi _{L})\cdot D\right] $ for every hidden proposition $L$.
\end{theorem}

\begin{proof}
$(\Longleftarrow )$ If $D=\sum_{k}w_{k}\cdot E_{\mathbb{C\cdot \psi }_{k}}$
is any density matrix (where $\left\{ \psi _{k}\right\} _{k\geq 1}$ is an
orthonormal basis of $\mathcal{H}$ and $\left\{ w_{k}\right\} _{k\geq 1}$ is
a sequence of real numbers in $\left[ 0,1\right] $ with $\sum_{k}w_{k}=1$)
then the examples above show that \ $\mu ^{\prime }(A)=\sum_{k}w_{k}\cdot
\eta _{\mathbb{C\cdot \psi }_{k}}(A\cap \mathbb{C\cdot \psi }_{k})$ (where $%
A $ is any pseudo-borel subset of $\mathcal{H}$) is a hidden mixed; since we
have $\mu ^{\prime }(L)=\sum_{k}w_{k}\cdot \left\langle \sigma (\chi
_{L})\right\rangle _{\psi _{k}}=Trace\left[ \sigma (\chi _{L})\cdot D\right]
=\mu (L)$ for every hidden proposition $L$ the measure $\mu $ is also a
hidden state.

$(\Longrightarrow )$ Since $\mu (L)=\mu (M)$ whenever $\sigma (\chi
_{L})=\sigma (\chi _{M})$ it is well defined a map $\widehat{\mu }:PR(%
\mathcal{H)\rightarrow }\left[ 0,1\right] $ on the space $PR(\mathcal{H)}$
of orthogonal projection operators by $\widehat{\mu }(E)=\mu (L)$ for $%
\sigma (\chi _{L})=E$. We will show that the map $\widehat{\mu }$ is a
measure on the closed subspaces of $\mathcal{H}$.

We need to prove that for every sequence $\left\{ E_{n}\right\} $ of
projectors with $E_{n}\cdot E_{m}=0$ whenever $n\neq m$ we have $\widehat{%
\mu }(\sum E_{n})=\sum \widehat{\mu }(E_{n})$.

Let $a_{n}=1-1/n$ for $n\geq 1$ and $a_{0}=-\infty $, consider the spectral
family defined by: $E(s)=0$ if $s<0$, $E(s)=E_{1}+\cdot \cdot \cdot +E_{n}$
if $a_{n}\leq s<a_{n+1}$ and $E(s)=I$ if $1\leq s$.

The self-adjoint operator $T$ defined by this family is bounded with $%
E_{]a_{n-1},a_{n}]}^{T}=E_{n}$ for every $n\geq 1$; let $f$ be an observable
function with $\sigma (f)=T$.

Let $L_{n}=f^{-1}]a_{n-1},a_{n}]$ for $n\geq 1$, every $L_{n}$ is a hidden
proposition with $\sigma (\chi _{L_{n}})=E_{n}$. Let $L=f^{-1}]-\infty ,1[$,
the set $L$ is a hidden proposition disjoint union of all the $\left\{
L_{n}\right\} $ with $\sigma (\chi _{L})=\sum_{n\geq 1}E_{n}$. Then $%
\widehat{\mu }(\sum_{n\geq 1}E_{n})=\mu (L)=\sum_{n\geq 1}\mu
(L_{n})=\sum_{n\geq 1}\widehat{\mu }(E_{n})$.

Therefore for the Gleason's Theorem (cfr. [G]) there exists a (unique)
density matrix $D$ such that $\widehat{\mu }(E)=Trace\left[ E\cdot D\right] $
for every projector $E$.
\end{proof}

We will denote by $\mathcal{S}$ the family of all hidden mixed states on $%
\mathcal{H}$ and by $DM(\mathcal{H})$ the set of all density matrices on $%
\mathcal{H}$, the previous theorem states there is a surjective map $\delta :%
\mathcal{S\rightarrow }DM(\mathcal{H})$ associating to a measure $\mu $ a
density matrix $\delta (\mu )$ such that $\mu (L)=Trace\left[ \sigma (\chi
_{L})\cdot \delta (\mu )\right] $ for every hidden proposition $L$.

\begin{theorem}
\bigskip For every hidden observable $f$ and every hidden mixed state $\mu $
we have: 
\begin{equation*}
Trace\left[ \sigma (f)\cdot \delta (\mu )\right] =\int f\cdot d\mu
\end{equation*}
\end{theorem}

\begin{proof}
Let's write $T=\sigma (f)$ and $D=\delta (\mu )$. For every real number $r$
the projector\ associated to the hidden proposition $f^{-1}\left( -\infty ,r%
\right] $ is $\sigma (\chi _{f^{-1}\left( -\infty ,r\right] })=\chi _{\left(
-\infty ,r\right] }\circ T=E_{\left( -\infty ,r\right] }^{T}$, therefore $%
\mu (f^{-1}\left( -\infty ,r\right] )=Trace\left[ E_{\left( -\infty ,r\right]
}^{T}\cdot D\right] $.

Remembering the properties of ch. 8 in [K and S] we have: $\int_{\mathcal{H}%
}f\cdot d\mu =\int_{\mathbb{R}}id\cdot df_{\ast }\mu =\int_{\mathbb{R}%
}id\cdot d\nu _{F}$ where $F(r)=f_{\ast }\mu \left( -\infty ,r\right] =\mu
(f^{-1}\left( -\infty ,r\right] )=Trace\left[ E_{\left( -\infty ,r\right]
}^{T}\cdot D\right] $.

When $D=E_{\mathbb{C\cdot \psi }_{k}}$ we get $F_{k}(r)=Trace\left[
E_{\left( -\infty ,r\right] }^{T}\cdot E_{\mathbb{C\cdot \psi }_{k}}\right]
=\left\langle E_{\left( -\infty ,r\right] }^{T}\right\rangle _{\psi _{k}}$
and $\int_{\mathcal{H}}f\cdot d\mu =\int_{\mathbb{R}}id\cdot d\nu
_{F_{k}}=\left\langle T\right\rangle _{\psi _{k}}=Trace\left[ T\cdot E_{%
\mathbb{C\cdot \psi }_{k}}\right] $.

For a general $D=\sum_{k}w_{k}\cdot E_{\mathbb{C\cdot \psi }_{k}}$ (where $%
\left\{ \psi _{k}\right\} _{k\geq 1}$ is an orthonormal basis of $\mathcal{H}
$ and $\left\{ w_{k}\right\} _{k\geq 1}$ is a sequence of real numbers in $%
\left[ 0,1\right] $ with $\sum_{k}w_{k}=1$) we have $F(r)=\sum_{k}w_{k}\cdot
\left\langle E_{\left( -\infty ,r\right] }^{T}\right\rangle _{\psi _{k}}$,
that is: $F=\sum_{k}w_{k}\cdot F_{k}$, and $\nu _{F}=\sum w_{k}\cdot \nu
_{F_{k}}$, therefore: $Trace\left[ T\cdot D\right] =\sum_{k}w_{k}\cdot Trace%
\left[ T\cdot E_{\mathbb{C\cdot \psi }_{k}}\right] =\sum_{k}w_{k}\cdot
\left\langle T\right\rangle _{\psi _{k}}=\sum_{k}w_{k}\cdot \int_{\mathbb{R}%
}id\cdot d\nu _{F_{k}}=\int_{\mathbb{R}}id\cdot d\nu _{F}$ and this proves
the equality.
\end{proof}

\begin{corollary}
\bigskip \bigskip For every hidden observable $f$ , every hidden mixed state 
$\mu $ and every $b$ in $\mathcal{B}$ we have : 
\begin{equation*}
Trace\left[ b\circ \sigma (f)\cdot \delta (\mu )\right] =\int b\circ f\cdot
d\mu
\end{equation*}
\end{corollary}

\begin{proof}
\bigskip Apply the previous theorem to $b\circ f$.
\end{proof}

\section{\protect\bigskip {\protect\Large A uniqueness theorem}}

\begin{center}
\bigskip
\end{center}

\begin{definition}
A theory with hidden variables (relative to a Hilbert space $\mathcal{H}$)
is given assigning:

\begin{itemize}
\item  a set $\Lambda $ (the hidden variables space)

\item  a surjective map $\pi :\Lambda \rightarrow \mathcal{P(H)}$ (defining
for each $\left[ \psi \right] $ in $\mathcal{P(H)}$ the fiber $\Lambda _{%
\left[ \psi \right] }=\pi ^{-1}(\left[ \psi \right] )$)

\item  on each fiber $\Lambda _{\left[ \psi \right] }$ a $\sigma -$algebra $%
\mathcal{M}_{\left[ \psi \right] }$ of subsets and a measure $\mu _{\left[
\psi \right] }$ making $\Lambda _{\left[ \psi \right] }$ a standard
non-atomic probability space

\item  a set $\mathcal{G}$ of real functions on $\Lambda $ (the hidden
observables) pseudo-measurables (that is $f^{-1}(B)\cap \Lambda _{\left[
\psi \right] }$ is a measurable subset of $\Lambda _{\left[ \psi \right] }$
for every $f$ in $\mathcal{G}$ and every borel subset $B$ in the real line)
and essentially bounded (that is each $f$ is bounded out of a suitable
subset $N_{f}$ of $\Lambda $ with $\mu _{\left[ \psi \right] }(N_{f}\cap
\Lambda _{\left[ \psi \right] })=0$ for each $\left[ \psi \right] $ in $%
\mathcal{P(H)}$)

\item  a surjective map $\beta :\mathcal{G\rightarrow }$ $\mathcal{B}_{sa}(%
\mathcal{H})$
\end{itemize}

such that: $\mu _{\left[ \psi \right] }(f^{-1}(B)\cap \Lambda _{\left[ \psi %
\right] })=\left\langle E_{B}^{\beta (f)}\right\rangle _{\psi }$ for every $%
f $ in $\mathcal{G}$, every borel subset $B$ in the real line and every $%
\left[ \psi \right] $ in $\mathcal{P(H)}$.
\end{definition}

\begin{itemize}
\item  Obviously the datum of $\mathcal{H}^{\cdot }=\mathcal{H\setminus }%
\left\{ 0\right\} $, of the canonical map $q:\mathcal{H}^{\cdot }\rightarrow 
\mathcal{P(H)}$, of the sets $\mathbb{C}_{\left[ \psi \right] }^{\cdot }=(%
\mathbb{C}\mathcal{\setminus }\left\{ 0\right\} )\cdot \psi $ with the
measures $\eta _{\left[ \psi \right] }$ and the set of functions $\mathcal{O}
$ with the map $\sigma $ defined in the previous sections is a hidden
variable theory

\item  For simplicity we consider on the fibers $\Lambda _{\left[ \psi %
\right] }$ the most natural structure of probability space; moreover the
functions in $\mathcal{G}$ are taken essentially bounded otherwise we would
need to deal with non-bounded self-adjoint operators.

\item  two pseudo-measurables functions on $\Lambda $ will be considered
equal if they coincide out of a pseudo-measurable subset of $\Lambda $.
\end{itemize}

\begin{definition}
Two pseudo-measurable and essentially bounded functions $f_{1}$ and $f_{2}$
on $\Lambda $ will be called statistically equivalent if $\mu _{\left[ \psi %
\right] }(f_{1}^{-1}(B)\cap \Lambda _{\left[ \psi \right] })=\mu _{\left[
\psi \right] }(f_{2}^{-1}(B)\cap \Lambda _{\left[ \psi \right] })$ for every
borel subset $B$ in the real line and every $\left[ \psi \right] $ in $%
\mathcal{P(H)}$; the family $\mathcal{G}$ of a hidden variable theory will
be called maximal if whenever $\mathcal{G}$ contains a function it contains
also all its statistically equivalent functions.
\end{definition}

\begin{itemize}
\item  if $f_{1}$ and $f_{2}$ in $\mathcal{G}\ $are statistically equivalent
then $\beta (f_{1})=\beta (f_{2})$

\item  the family $\mathcal{G}$ can always be extended to a maximal family $%
\widetilde{\mathcal{G}}$ and the map $\beta $ can be extended in a unique
way to a map $\widetilde{\beta }$ in such a way to have the same value on
statistically equivalent functions. Considering this new family $\widetilde{%
\mathcal{G}}$ instead of $\mathcal{G}$ and the map $\widetilde{\beta }$
instead of $\beta $ we get a new hidden variables theory.
\end{itemize}

\begin{theorem}
The family $\mathcal{O}$ is maximal.
\end{theorem}

\begin{proof}
Let $f$ be a function in $\mathcal{O}$ and let $g$ be a (pseudo-measurable
and essentially bounded) function on $\mathcal{H}^{\cdot }$ statistically
equivalent to $f$; to prove that $g$ is also in $\mathcal{O}$ we have to
show that $g$ has horthodox mean values and that each $g^{-1}(B)$ is a
hidden proposition for every borel subset $B$ in the real line.

Let $T=\sigma (f)$, fixed $\psi $ in $\mathcal{H}^{\cdot }$ let's define $F:%
\mathbb{R\rightarrow }\left[ 0,1\right] $ by $F(r)=\left\langle E_{(-\infty
,r]}^{T}\right\rangle _{\psi }$. Since $T$ is bounded the measure $\nu _{F}$
has support inside a suitable bounded interval$.$ Denoted by $f|$ and $g|$
the restrictions of $f$ and $g$ to $\mathbb{C}_{\left[ \psi \right] }^{\cdot
}$ we have $g|_{\ast }\eta _{\left[ \psi \right] }=f|_{\ast }\eta _{\left[
\psi \right] }=\nu _{F}$; because the identity function in $\mathbb{R}$ is
absolutely integrable with respect to $\nu _{F}$ the function $g|$ is
absolutely integrable with respect to $\eta _{\left[ \psi \right] }$ (cfr.
Cor. 3 pag. 93 of [K and S]) and $\ \ \int_{\mathbb{C}_{\left[ \psi \right]
}^{\cdot }}g|$ $\cdot d\eta _{\left[ \psi \right] }=\int_{\mathbb{R}}id\cdot
d\nu _{F}=\int_{\mathbb{R}}id\cdot d\nu _{\left\langle E_{(-\infty
,r]}^{T}\right\rangle _{\psi }}=\left\langle T\right\rangle _{\psi }$. This
proves that $g$ has horthodox mean values.

Fixed a borel subset $B$ in the real line since the characteristic function $%
\chi _{g^{-1}(B)}$ has mean values: $\int_{\mathbb{C}_{\left[ \psi \right]
}^{\cdot }}\chi _{g^{-1}(B)}\cdot d\eta _{\left[ \psi \right] }=\eta _{\left[
\psi \right] }(f^{-1}(B)\cap \mathbb{C}_{\left[ \psi \right] }^{\cdot
})=\left\langle E_{B}^{T}\right\rangle _{\psi }$ given by the projector $%
E_{B}^{T}$ the set $g^{-1}(B)$ is a hidden proposition.
\end{proof}

\bigskip

\begin{definition}
An isomorphism (mod $0$) between two hidden variables theories \ \ \ \ \ ($%
\Lambda $, $\pi $, $\mu _{\cdot }$, $\mathcal{G}$, $\beta $) and ($\Lambda
^{\prime }$, $\pi ^{\prime }$, $\mu _{\cdot }^{\prime }$, $\mathcal{G}%
^{\prime }$, $\beta ^{\prime }$) is given by a map$\ \Phi :\Lambda \setminus
N\rightarrow \Lambda ^{\prime }\setminus N^{\prime }$(where $N$ and $%
N^{\prime }$ are pseudo-measurable null subsets respectively of $\Lambda $
and $\Lambda ^{\prime }$) with the following properties:

\begin{itemize}
\item  $\Phi $ is bijective

\item  $\pi =\pi ^{\prime }\circ \Phi $ (and therefore $\Phi (\Lambda _{%
\left[ \psi \right] })\subset \Lambda _{\left[ \psi \right] }^{\prime }$ for
every $\left[ \psi \right] $ in $\mathcal{P(H)}$)

\item  $\Phi |:\Lambda _{\left[ \psi \right] }\setminus N^{{}}\rightarrow
\Lambda _{\left[ \psi \right] }^{\prime }\setminus N^{\prime }$ is a measure
preserving borel equivalence ($\Phi |_{\ast }\mu _{\left[ \psi \right] }=\mu
_{\left[ \psi \right] }^{\prime }$) for every $\left[ \psi \right] $ in $%
\mathcal{P(H)}$

\item  $\Phi _{\ast }\mathcal{G\subset G}^{\prime }$ (where $\Phi _{\ast
}(f)=f\circ \Phi ^{-1}$) and $\Phi _{\ast }:\mathcal{G\rightarrow G}^{\prime
}$ is a bijective map
\end{itemize}
\end{definition}

\bigskip

Note that an isomorphism (mod $0$) $\Phi :\Lambda \setminus N\rightarrow
\Lambda ^{\prime }\setminus N^{\prime }$ automatically verifies the
condition: $\beta =\beta ^{\prime }\circ \Phi _{\ast }$. In fact taken $f $
in $\mathcal{G}$ let $f^{\prime }=\Phi _{\ast }(f)$, $T=\beta (f)$ and $%
T^{\prime }=\beta ^{\prime }(f^{\prime })$ we have: $\left\langle
E_{B}^{T}\right\rangle _{\psi }=\mu _{\left[ \psi \right] }(f^{-1}(B)\cap
\Lambda _{\left[ \psi \right] })=\mu _{\left[ \psi \right] }^{\prime
}((f^{\prime })^{-1}(B)\cap \Lambda _{\left[ \psi \right] }^{\prime
})=\left\langle E_{B}^{T^{\prime }}\right\rangle _{\psi }$ for every $\left[
\psi \right] $ in $\mathcal{P(H)}$ and every borel subset $B$ in the real
line, therefore $T=T^{\prime }$.

\bigskip

\begin{theorem}
\bigskip Two hidden variables theories with maximal spaces of hidden
observables are isomorphic (mod $0$).
\end{theorem}

\begin{proof}
For each $\left[ \psi \right] $ in $\mathcal{P(H)}$ since $\Lambda _{\left[
\psi \right] }$ and $\Lambda _{\left[ \psi \right] }^{\prime }$ are standard
non-atomic probability spaces there exists an isomorphism (mod $0$) (a
measure preserving borel equivalence) $\Phi _{\left[ \psi \right] }:\Lambda
_{\left[ \psi \right] }\setminus N_{\left[ \psi \right] }^{{}}\rightarrow
\Lambda _{\left[ \psi \right] }^{\prime }\setminus N_{\left[ \psi \right]
}^{\prime }$(where $N_{\left[ \psi \right] }^{{}}$ and $N_{\left[ \psi %
\right] }^{\prime }$ are measurable null subsets respectively of $\Lambda _{%
\left[ \psi \right] }$and $\Lambda _{\left[ \psi \right] }^{\prime }$).
Taken $N=\bigcup N_{\left[ \psi \right] }^{{}}$ and $N^{\prime }=\bigcup N_{%
\left[ \psi \right] }^{\prime }$ and defined $\Phi :\Lambda \setminus
N\rightarrow \Lambda ^{\prime }\setminus N^{\prime }$ by $\Phi (\lambda
)=\Phi _{\left[ \psi \right] }(\lambda )$ when $\lambda $ is in $\Lambda _{%
\left[ \psi \right] }\setminus N_{\left[ \psi \right] }^{{}}$ we get a map
veryfying the first three conditions of an isomorphism (mod $0$)$.$

Taken $f$ in $\mathcal{G}$the function $f^{\prime }=\Phi _{\ast }(f)$ is
pseudo-measurable and essentially bounded on $\Lambda ^{\prime }$,
considered $T=\beta (f)$ and choosen $g^{\prime }$in $\mathcal{G}^{\prime }$
such that $\beta ^{\prime }(g^{\prime })=T$ \ let's prove that $f^{\prime }$
is statistically equivalent to $g^{\prime }$, the maximality of $\mathcal{G}%
^{\prime }$ will imply then that also $f^{\prime }$ is in $\mathcal{G}%
^{\prime }$.

We have in fact: $\mu _{\left[ \psi \right] }^{\prime }((g^{\prime
})^{-1}(B)\cap \Lambda _{\left[ \psi \right] }^{\prime })=\left\langle
E_{B}^{T^{\prime }}\right\rangle _{\psi }=\mu _{\left[ \psi \right]
}(f^{-1}(B)\cap \Lambda _{\left[ \psi \right] })=\mu _{\left[ \psi \right]
}^{\prime }((f^{\prime })^{-1}(B)\cap \Lambda _{\left[ \psi \right]
}^{\prime })$ for every $\left[ \psi \right] $ in $\mathcal{P(H)}$ and every
borel subset $B$ in the real line. For the generality of $f$ this means: $%
\Phi _{\ast }\mathcal{G\subset G}^{\prime }$, in an analogous way we can
prove that $(\Phi ^{-1})_{\ast }\mathcal{G}^{\prime }\mathcal{\subset G}$.
Since $\Phi _{\ast }\circ (\Phi ^{-1})_{\ast }=id_{\mathcal{G}^{\prime }}$
and $(\Phi ^{-1})_{\ast }\circ \Phi _{\ast }=id_{\mathcal{G}}$ the map $\Phi
_{\ast }$ is bijective.
\end{proof}

\bigskip

Therefore every hidden variables theory is isomorphic (mod $0$) to the
theory \ \ \ \ \ \ \ \ \ ( $\mathcal{H}^{\cdot }$, $q$, $\eta _{\cdot }$, $%
\mathcal{O}$, $\sigma $) developped in the previous sections if its space of
hidden observables is maximal otherwise it is isomorphic (mod $0$) to a
theory ( $\mathcal{H}^{\cdot }$, $q$, $\eta _{\cdot }$, $\mathcal{O}^{\prime
}$, $\sigma $) with $\mathcal{O}^{\prime }\subset \mathcal{O}$.

\bigskip

\begin{center}
\bigskip {\Large Bibliography:}
\end{center}

[B] G. Boniolo(a cura di): Filosofia della fisica.

\ \ \ \ \ \ Bruno Mondadori, Milano 1997

\bigskip

[C] \ \ A. Cassa: Quantum physical systems as classical systems.

\ \ \ \ \ \ \ \ \ J. Math. Phys. -2001-vol. 42, N. 11- p. 5143-49

\bigskip

[F1] A. Fine, Hidden Variables, Joint Probability, and the Bell Inequalities

\ \ \ \ \ \ \ \ Phys. Rev. Lett. 48, 291 (1982)

\bigskip

[F2] A. Fine: Joint distributions quantum correlations and commuting
observables

\ \ \ \ \ \ \ \ J. Math. Phys. 23(7) 1306-1310

\bigskip

[G] \ A.M. Gleason: Measures on the closed subspaces of a Hilbert space.

\ \ \ \ \ \ \ \ J. Math. and Mech. -1957- N.6, p. 123-133

\bigskip

[G-D] G.C. Ghirardi, F. De Stefano: Il mondo quantistico, una realt\`{a}
ambigua. in Ambiguit\`{a} (a cura di) C. Magris

\ \ \ \ \ \ \ \ \ \ \ Longo, Moretti e Vitali Bergamo 1996

\bigskip

[J-P] J. M. Jauch and C. Piron: Can hidden variables be excluded in quantum
mechanics?

\ \ \ \ \ \ \ \ Helv. Phys. Acta 36, 827-837 (1963)

\bigskip

[K-S] S. Kochen \& E.P. Specker: The problem of hidden variables in quantum
mechanics.

\ \ \ \ \ \ \ \ \ \ J. Math. Mechan.17 (1967) 5987

\bigskip

[K and S] \ \ J.L. Kelley and T.P. Srinivasan: Measure and Integral - Vol. 1.

\ \ \ \ \ \ \ \ \ \ \ \ \ \ \ \ \ \ Springer-Verlag, NY 1988

\bigskip

[M1] N. D. Mermin, Simple Unified Form for the Major No-Hidden-Variables
Theorems.

\ \ \ \ \ \ \ \ \ Phys. Rev. Lett. 65 3373 (1990)

\bigskip

[M2] N. D. Mermin, Hidden variables and the two theorems of John Bell.

\ \ \ \ \ \ \ \ \ Rev. Mod. Phys. 65 803 (1993)

\bigskip

[P] A.Peres: Two simple proofs of the Kochen Specker theorem.

\ \ \ \ \ \ J.Phys. A (math. gen.) 24,4 (1991)

\bigskip

[V] \ \ V.S. Varadarajan: Geometry of quantum theory. Vol. I

\ \ \ \ \ \ \ \ Van Nostrand, Princeton 1984

\bigskip

[VN] J. von Neumann: Mathematical foundations of quantum mechanics.

\ \ \ \ \ \ \ \ \ Princeton univ. press, Princeton NJ 1955

\bigskip

[W] \ \ J. Weidmann: Linear operators in Hilbert spaces.

\ \ \ \ \ \ \ \ \ Springer-Verlag, NY 1980

\bigskip

\end{document}